\documentclass[twocolumn]{revtex4}
  
\usepackage{amsmath}
\usepackage{amssymb}
\usepackage{url}
\usepackage{graphicx}

\usepackage{amsthm}
\usepackage{enumerate}
\newtheorem{theorem}{Theorem} %[section]
\newtheorem{lemma}[theorem]{Lemma}

\theoremstyle{plain}

\newcommand*{\cE}{\mathcal{E}}
\newcommand*{\cF}{\mathcal{F}}
\newcommand*{\cG}{\mathcal{G}}
\newcommand*{\cH}{\mathcal{H}}

\newcommand*{\cK}{\mathcal{K}}
\newcommand*{\cL}{\mathcal{L}}

\newcommand*{\cN}{\mathcal{N}}

\newcommand*{\cR}{\mathcal{R}}

\newcommand*{\cS}{\mathcal{S}}

\newcommand*{\cT}{\mathcal{T}}
\newcommand*{\cV}{\mathcal{V}}
\newcommand*{\cW}{\mathcal{W}}

\newcommand*{\cEp}{\mathcal{E'}}
\newcommand*{\cFp}{\mathcal{F'}}

\newcommand{\Sym}{\mbox{Sym}}

\newcommand*{\eps}{\varepsilon}
\newcommand*{\tr}{\mathsf{tr}}
\newcommand*{\id}{\mathsf{id}}

\newcommand*{\ket}[1]{|#1\rangle}
\newcommand*{\bra}[1]{\langle #1|}
\newcommand*{\proj}[1]{\ket{#1}\bra{#1}}

\newcommand{\be}{\begin{equation}}
\newcommand{\ee}{\end{equation}}
\newcommand{\bea}{\begin{eqnarray}}
\newcommand{\eea}{\end{eqnarray}}
\newcommand{\bestar}{\begin{equation*}}
\newcommand{\eestar}{\end{equation*}}
\newcommand{\beastar}{\begin{eqnarray*}}
\newcommand{\eeastar}{\end{eqnarray*}}

\newcommand*{\End}{\mathrm{End}}

\newcommand*{\dia}{\diamond}

\begin{document}

\title{Post-selection technique for quantum channels with applications
  to quantum cryptography}

\author{Matthias \surname{Christandl}}
%\email[]{christandl@lmu.de} 
\affiliation{Faculty of Physics, Ludwig-Maximilians-University
Munich, 80333 Munich, Germany}

\author{Robert \surname{K{\"o}nig}}
%\email[]{rkoenig@caltech.edu} 
\affiliation{Institute for Quantum Information, California Institute of Technology, Pasadena, CA 91125, USA}

\author{Renato \surname{Renner}}
%\email[]{renner@phys.ethz.ch} 
\affiliation{Institute for Theoretical Physics, ETH Zurich, 8093 Zurich, Switzerland}

\begin{abstract}
  We propose a general method for studying properties of quantum
  channels acting on an $n$-partite system, whose action is invariant
  under permutations of the subsystems. Our main result is that, in
  order to prove that a certain property holds for any arbitrary
  input, it is sufficient to consider the special case where the input
  is a particular \emph{de Finetti-type state}, i.e., a state which
  consists of $n$ identical and independent copies of an (unknown)
  state on a single subsystem.  A similar statement holds for more
  general channels which are covariant with respect to the action of
  an arbitrary finite or locally compact group.

  Our technique can be applied to the analysis of
  information-theoretic problems. For example, in quantum
  cryptography, we get a simple proof for the fact that security of a
  discrete-variable quantum key distribution protocol against
  collective attacks implies security of the protocol against the most
  general attacks. The resulting security bounds are tighter than
  previously known bounds obtained by proofs relying on the
  exponential de Finetti theorem~\cite{Renner07}.
\end{abstract}

%2.20.Qs General properties, structure, and representation of Lie groups
%03.67.-a Quantum information
%03.67.Dd Quantum cryptography
\date{September 17, 2008}

\maketitle

\newcommand*{\stateset}{\cS}
\newcommand*{\qual}{\lambda}
\newcommand*{\paramset}{\Lambda}

\pagestyle{plain}
In quantum mechanics, the most general way of describing the evolution of a subsystem $A$ ($A$ may be part of a larger system) at time $t$  to a subsystem $A'$ at a later point in time $t'$ is by application of a \emph{quantum channel}. Mathematically, a quantum channel is a \emph{completely positive trace-preserving (CPTP) map} transforming the reduced density matrix $\rho_A$ of system $A$ at time $t$  to $\rho_{A'}$, the reduced density matrix of system $A'$ at time $t'$.
CPTP maps are used in various areas of physics
and information theory. A CPTP map modeling a particular quantum communication channel, for instance, describes how the channel output
$\rho_{A'}$ depends on the input $\rho_A$.

A common method to characterize a given CPTP map $\cE$ is to compare it
to an \emph{idealized} CPTP map $\cF$ that is well understood, e.g.,
because it has a simple description. For instance, given a physical
communication channel specified by $\cE$, one may characterize its
ability to reliably transmit messages by showing its similarity to a
perfect channel $\cF$ characterized by the identity mapping
$\id$. Another example is the analysis of information-theoretic or
cryptographic \emph{protocols} (e.g., for quantum key
distribution). Here, $\cE$ may be the action of the \emph{actual protocol} while $\cF$ is the \emph{ideal functionality} the
protocol is supposed to reproduce. We are then typically interested in
proving that $\cE$ is almost equal to $\cF$ (in quantum cryptography, this corresponds to proving security).

In order to compare two CPTP maps $\cE$ and $\cF$, we need a notion of
distance. A natural choice is the metric induced by the diamond norm
$\| \cdot \|_{\dia}~$\footnote{The diamond norm is given by $\| \cE
  \|_{\dia}=\sup_{k \in \mathbb{N}} \| \cE \otimes \id_k\|_1$ where
  $\| \cF \|_{1}:= \sup_{\|\sigma\|_1\leq 1} \|\cF(\sigma) \|_1$ and
  $\| \sigma\|_1:=\tr \sqrt{\sigma^\dagger \sigma}$ is the trace
  norm. $\id_k$ denotes the identity map on states of a
  $k$-dimensional quantum system. The suprema are reached for positive
  $\sigma$ and $k$ equal to the dimension of the input of
  $\cE$~\cite{Kitaev97}.} since it is directly related to
the maximum probability that a difference can be observed between the
processes described by $\cE$ and $\cF$, respectively. More precisely,
consider a hypothetical game where a player is asked to guess whether
a given physical process is described by $\cE$ or $\cF$, which are
both equally likely to be the correct descriptions. If the player is
allowed to observe the process once (with an input of his choice,
possibly correlated with a reference system) then the maximum
probability $p$ of a correct guess is given by $ p = \frac{1}{2} +
\frac{1}{4} \| \cE - \cF \|_{\dia} $ .  In particular, if $\cE$ and
$\cF$ are identical, the distance equals zero and, hence, $p =
\frac{1}{2}$, corresponding to a random guess. On the other hand, if
$\cE$ and $\cF$ are perfectly distinguishable, we have $\| \cE - \cF
\|_{\dia} = 2$ and $p = 1$.

Here, we present a general method for computing an upper bound on the
distance $\|\cE - \cF \|_{\dia}$ between two maps $\cE$ and $\cF$,
provided they act symmetrically on an
$n$-partite system with subsystems $\cH$ of finite
dimension.
While, by
definition, the diamond norm involves a maximization over all possible
inputs, we show that for calculating the bound it is sufficient to
consider (relative to a reference system) the particular input
\begin{align} \label{eq:deFinstate}
\tau_{\cH^n} = \int \sigma_{\cH}^{\otimes n} \mu(\sigma_{\cH}) \ ,
\end{align}
where $\mu(\cdot)$ is the measure on the space of density operators on
a single subsystem induced by the Hilbert-Schmidt metric. States of
the form~\eqref{eq:deFinstate} are also known as \emph{de Finetti
  states}. They describe the joint state of $n$ subsystems prepared as
identical and independent copies of an (unknown) density operator
$\sigma_{\cH}$. Because of their structure, de Finetti states are
usually easy to handle in calculations and proofs, as outlined below.

As an example, we apply this result to the security analysis of
quantum key distribution (QKD) schemes~\cite{BenBra84,Ekert91}. Let
$\cE$ be the map describing a given QKD protocol, which takes as input
$n$ predistributed particle pairs (which may have been generated in a
preliminary protocol step). Security of the protocol (against the most
general attacks) is then defined by the requirement that the
\emph{protocol} $\cE$ is close to the \emph{ideal functionality} $\cF$
that simply outputs a \emph{perfect key}, independently of the input
(which may be arbitrarily compromised by the action of an adversary).
Now, according to our main result, this distance is bounded by simply
evaluating the map $\cE$ for an input of the
form~\eqref{eq:deFinstate} and comparing the generated key with a
perfect key. We further show that this result is equivalent to proving
security of the scheme against a restricted type of attacks, called
\emph{collective attacks}, where the adversary is assumed to attack
each of the particle pairs independently and identically.  Our result
thus gives a simple proof for the statement (proved originally
in~\cite{Renner05, Renner07}) that security of a QKD protocol against
collective attacks implies security against the most general
attacks. The resulting security bounds are tighter than previously
known bounds obtained by proofs relying on the exponential de Finetti
theorem~\cite{Renner07}.

\emph{Main Result.}  Let $\Delta$ be a linear map from
$\End(\cH^{\otimes n})$ to $\End(\cH')$. In particular, $\Delta$ may
be the difference between two CPTP maps. $\End(\cL )$ denotes the
space of all endomorphisms on $\cL$, which includes the density
operators on $\cL$. We denote by $\pi$ the map on $\End(\cH^{\otimes
  n})$ that permutes the subsystems with permutation
$\pi$~\footnote{The permutation $\pi $ on $n$ elements acts on
  $\cH^n=\cH^{\otimes n}$ by permuting the tensor factors, i.e., $\pi
  \ket{i_1 \cdots i_n}= \ket{i_{\pi^{-1}(1)} \cdots i_{\pi^{-1}(n)}} $
  for a basis $\{\ket{i}\}$ of $\cH$. The space of vectors invariant
  under the action of all $\pi$ is denoted by $\Sym^n(\cH)$. As a map
  on $\End(\cH^{\otimes n})$ we write $\pi(\rho)=\pi \rho \pi^{-1}$.}.
Our main result, the \emph{Post-Selection Theorem}~\footnote{The proof
  essentially relies on the fact that the state that maximizes the
  diamond norm can be \emph{post-selected} by a measurement as
  constructed in Lemma~\ref{lem:extractpart}.}, gives an upper bound
on the norm of a permutation-invariant map in terms of the action of
the map on a purification~\footnote{A \emph{purification}
  $\tau_{\cH^n \cR}$ of $\tau_{\cH^n}$ is a pure state on
  $\cH^{n} \otimes \cR$ satisfying $\tr_{\cR} \tau_{\cH^n\cR}=
  \tau_{\cH^n}$} $\tau_{\cH^n\cR}$ of the state $\tau_{\cH^n}$ defined
by~\eqref{eq:deFinstate}.

\begin{theorem} \label{thm:main} If for any permutation $\pi$ there
  exists a CPTP map $\cK_{\pi}$ such that $\Delta \circ \pi =
  \cK_{\pi} \circ \Delta$,
then
\begin{align*}
  \| \Delta \|_{\dia} \leq g_{n,d} \bigl\| (\Delta \otimes \id)(\tau_{\cH^n\cR}) \|_1 \ .
\end{align*}
$\id$ denotes the identity
map on $\End(\cR)$ and $g_{n,d} = {\textstyle
  \binom{n+d^2-1}{n} }\leq ({n+1})^{d^2-1}$, for $d = \dim \cH$.
\end{theorem}

The proof of Theorem~\ref{thm:main} uses the following lemma which
relates arbitrary density operators $\rho_{\cH^n\cK^n}$ on the
symmetric subspace $\Sym^n(\cH \otimes \cK)\subset (\cH\otimes\cK)^{\otimes n}$, for $\cK \cong \cH$, to
a particular purification of $\tau_{\cH^n}$. We define the state
$\tau_{\cH^n\cK^n} = \int \sigma_{\cH\cK}^{\otimes n}d(\sigma_{\cH
  \cK}) $ on $\Sym^n(\cH \otimes \cK)$, where $d(\cdot)$ is the
measure on the pure states induced by the Haar measure
on the unitary group acting on $\cH\otimes\cK$. We note that
$\tau_{\cH^n\cK^n}$ extends the state $\tau_{\cH^n}$ defined
in~\eqref{eq:deFinstate}, i.e., $\tr_{\cK^n}
\tau_{\cH^n\cK^n}= \tau_{\cH^n}$; the measure $\mu(\cdot )$
furthermore is the one induced by the Hilbert-Schmidt metric on
$\End(\cH)$~\cite{ZycSom01}. Let now $\tau_{\cH^n\cK^n\cN} $ be a
purification of $\tau_{\cH^n\cK^n}$.

\begin{lemma} \label{lem:extractpart} For $\rho_{\cH^n\cK^n}$ a
 density operator supported on $\Sym^n({\cH \otimes \cK})$, with $\cK \cong \cH$, there exists a
 trace-non-increasing map $\cT$ from the purifying system $\End(\cN)$
 to $\mathbb{C}$ such that
\begin{align} \label{eq:substate}
  \rho_{\cH^n\cK^n} = g_{n, d} \, (\id \otimes \cT)(\tau_{\cH^n\cK^n\cN}) \ ,
\end{align}
where $\id$ is the identity map on $\End((\cH \otimes
\cK)^{\otimes n})$ and $d=\dim\cH$.
\end{lemma}

\begin{proof}
  Let $\cN \cong \Sym^n(\cH \otimes \cK)$ and let $\{\ket{\nu_i}\}_i$
  be an eigenbasis of $\rho_{\cH^n\cK^n}$. Since, by Schur's lemma,
  ${\tau}_{\cH^n\cK^n}$ is the state proportional to the identity on
  $\Sym^n(\cH \otimes \cK)$, ${\tau}_{\cH^n\cK^n\cN}:=
  \proj{\Psi}_{\cH^n\cK^n\cN}$ is a purification of
  $\tau_{\cH^n\cK^n}$, where $\ket{\Psi}_{\cH^n\cK^n\cN} :=
  g_{n,d}^{-\frac{1}{2}} \sum_i \ket{\nu_i} \otimes \ket{\nu_i}$, and
  $g_{n,d}$ is the dimension of $\Sym^n(\cH \otimes
  \cK)$. Furthermore, for any basis vector $\ket{\nu_i}$,
\begin{align*}
  \proj{\nu_i}_{\cH^n\cK^n}
=
  g_{n, d} \, \tr_{\cN}(\tau_{\cH^n\cK^n\cN} \cdot \openone_{\cH^n\cK^n} \otimes \proj{\nu_i}_{\cN}) \ ,
\end{align*}
where  $\openone_{\cH^n\cK^n}\in\End((\cH\otimes\cK)^{\otimes n})$ is the identity.
This implies~\eqref{eq:substate} with $\cT: \sigma \mapsto \tr(\sigma
\rho_{\cN})$, since $\{\ket{\nu_i}\}_i$ is an eigenbasis of
$\rho_{\cH^n\cK^n}$. Because $\cT$ is clearly
trace-non-increasing, this concludes the proof.
\end{proof}

\begin{proof}[Proof of Theorem~\ref{thm:main}]
 We need to show that for any finite-dimensional space $\cR'$ and any density
 operator $\rho_{\cH^n\cR'}$,
\begin{align} \label{eq:fdist}
  \bigl\| (\Delta \otimes \id)(\rho_{\cH^n\cR'}) \bigr\|_1
\leq
  g_{n,d} \bigl\| (\Delta \otimes \id)(\tau_{\cH^n\cR}) \bigr\|_1 \ ,
\end{align}
for some purification $\tau_{\cH^n\cR}$ of $\tau_{\cH^n}$.
In a first step, we show that it is sufficient to
prove~\eqref{eq:fdist} for density operators $\rho_{\cH^n\cR'}$ with support on $\Sym^n({\cH \otimes
  \cK})$, where $\cK \cong \cH$ and $\cR' = \cK^{\otimes n}$. To see
this, let $\rho_{\cH^n\cR'}$ be an arbitrary density operator and define the density operator
\begin{align*}
  \bar{\rho}_{\cH^n\cR'\cR''}
=
  \frac{1}{n!} \sum_{\pi} (\pi \otimes \id)(\rho_{\cH^n\cR'}) \otimes \proj{\pi}_{\cR''} \ ,
\end{align*}
where the sum ranges over all permutations $\pi$ of the $n$
subsystems and where $\{\ket{\pi}\}_{\pi}$ is an orthonormal
family of vectors on an auxiliary space $\cR''$. Then, by
construction, the reduced state $\bar{\rho}_{\cH^n} =\tr_{\cR' \cR''}(\bar{\rho}_{\cH^n\cR' \cR''})$ is permutation invariant. Hence, according to~\cite{Renner05,ChrKoeMitRen07}, there exists a purification $\bar{\rho}_{\cH^n\cK^n}$ of $\bar\rho_{\cH^n}$
supported on $\Sym^n({\cH \otimes \cK})$. In
particular, because all purifications are equivalent up to
isometries, there exists a CPTP map $\cG$ from $\End(\cK^{\otimes n})$
to $\End(\cR' \otimes \cR'')$ such that $\bar{\rho}_{\cH^n\cR'\cR''} = (\id \otimes
\cG)(\bar{\rho}_{\cH^n\cK^n})$.  Making use of the assumption on the permutation
invariance of $\Delta$, we thus find that $\bigl\| (\Delta \otimes  \id)(\rho_{\cH^n\cR'}) \bigr\|_1$ equals 
\begin{align*}
  \frac{1}{n!}& \sum_{\pi} \bigl\| \bigl((\Delta \circ \pi) \otimes \id \bigr)(\rho_{\cH^n\cR'}) \bigr\|_1  =
  \bigl\| (\Delta \otimes \id)(\bar{\rho}_{\cH^n\cR'\cR''}) \bigr\|_1 \\
 &=
  \bigl\| (\Delta \otimes \cG)(\bar{\rho}_{\cH^n\cK^n}) \bigr\|_1  \leq
  \bigl\| (\Delta \otimes \id)(\bar{\rho}_{\cH^n\cK^n}) \bigr\|_1 \ ,
\end{align*}
where the last inequality holds because a CPTP map cannot increase the
norm.  It thus remains to show that~\eqref{eq:fdist} holds for states
$ \rho_{\cH^n\cK^n}$ in $\Sym^n({\cH \otimes \cK})$. By
Lemma~\ref{lem:extractpart} there exists a map $\cT$ such that
$\rho_{\cH^n\cK^n} = g_{n,d} \, ({\id \otimes
  \cT})(\tau_{\cH^n\cK^n\cN})$. Then, by linearity, we have
\begin{align*}
  \bigl\| (\Delta \otimes \id)(\rho_{\cH^n\cK^n}) \bigr\|_1
=
  g_{n,d} \bigl\| ( \Delta \otimes \cT)(\tau_{\cH^n\cK^n\cN}) \bigr\|_1 \ .
\end{align*}
Inequality~\eqref{eq:fdist} then follows from the fact that $\cT$
cannot increase the norm and by setting $\cR=\cK^{\otimes n} \otimes
\cN$.
\end{proof}

\emph{Application to Quantum Key Distribution.} QKD is the art of
generating a \emph{secret key} known only to two distant parties,
\emph{Alice} and \emph{Bob}, connected by an insecure quantum
communication channel and an authentic classical
channel~\footnote{\emph{Authenticity} means that the communication
  cannot be altered by an adversary. If only completely insecure
  channels are available, authenticity may be simulated using a short
  initial key shared between Alice and Bob.}.  Most QKD protocols can
be subdivided into two parts. In the first, Alice and Bob use the
quantum channel to distribute $n$ entangled particle pairs (this phase
may include advanced quantum protocols such as quantum repeaters). In
the second part, they apply local measurements (we will restrict
ourselves to the typical case of measurements that are independent and
identical on each of the $n$ pairs) followed by a sequence of
classical post-processing steps (such as \emph{parameter estimation},
\emph{error correction}, and \emph{privacy amplification}) to extract
$\ell$ key bits~\footnote{This describes an \emph{entanglement-based}
  protocol. However, our results immediately extend to
  \emph{prepare-and-measure schemes}, because their security analysis
  can generally be reduced to corresponding entanglement-based
  schemes~\cite{BeBrMe92}.}. It induces a map $\cE$ from $(\cH_A
\otimes \cH_B)^{\otimes n}$ (the $n$ particle pairs) to the set of
pairs $(S_A, S_B)$ of $\ell$-bit strings (Alice and Bob's final keys,
respectively) and $C$, where $C$ is a transcript of the classical
communication. Note that $\ell$ may depend on the input; in
particular, $\ell = 0$ if the entanglement of the initial particle
pairs is too small for key extraction.

A QKD protocol is said to be \emph{$\eps$-secure} (for some small
$\eps \geq 0$) if, for any attack of an adversary, the final keys
$S_A$ and $S_B$ computed by Alice and Bob are identical, uniformly
distributed, and independent of the adversary's knowledge, except with
probability $\eps$. This criterion can be reformulated as a condition
on the map $\cE$. Since an adversary may have full control over the
quantum channel connecting Alice and Bob, we require that, for
\emph{any} input to $\cE$, the output is a pair $(S_A, S_B)$ of secure
keys of length $\ell \geq 0$~\footnote{Of course, any non-trivial
  protocol generates keys of positive length $\ell > 0$ for at least
  \emph{some} inputs.}. To make this more precise, let $\cS$ be the
map that acts on the output $(S_A, S_B, C)$ of $\cE$ by replacing
$(S_A, S_B)$ by a pair $(S'_A, S'_B)$ of identical and uniformly
distributed keys of the same length, while leaving $C$ unchanged. With
this definition, the concatenated map $ \cF :=\cS \circ \cE$ describes
an \emph{ideal} key distillation scheme which always outputs a perfect
key pair. We then say that $\cE$ is \emph{$\eps$-secure} if $\| \cE -
\cF \|_{\dia} \leq \eps$.

$\cE$ is typically
invariant under permutations of the inputs. However, if it is not,
permutation invariance can be enforced by prepending an additional
\emph{symmetrization step} where both Alice and Bob permute their
inputs according to a permutation $\bar{\pi}$ chosen at random by one
party and communicated to the other using the classical channel~\footnote{It is
easy to see that this symmetrized key distillation protocol $\cE$
satisfies $\cK_{\pi} \circ \cE \circ \pi = \cE$ for any permutation
$\pi$, where $\cK_{\pi}$ is the operation that acts on the output
$(S_A, S_B, C)$ by replacing the communicated permutation $\bar{\pi}$
(in $C$) by $\bar{\pi} \circ \pi$. Similarly, we have $\cK_{\pi} \circ
(\cS \circ \cE) \circ \pi = \cS \circ \cE$, because $\cK_{\pi}$ acts
like the identity on the key pair $(S_A, S_B)$.}. We can thus apply
Theorem~\ref{thm:main} with $\Delta := \cE-\cF$, which
implies that $\cE$ is $\eps$-secure whenever
\begin{align} \label{eq:redcrit}
 \bigl\|\bigl((\cE - \cF) \otimes \id \bigr)(\tau_{\cH^n \cR})\bigr\|_1
\leq
 \bar{\eps} := \eps (n+1)^{-(d^2-1)} \ ,
\end{align}
where $\cH := {\cH_A \otimes \cH_B}$, where $d = \dim(\cH)$, and where
$\tau_{\cH^n \cR}$ is a purification of the state $\tau_{\cH^n}$
defined by~\eqref{eq:deFinstate}.

We will now employ~\eqref{eq:redcrit} to show that for proving
security of a QKD protocol it suffices to consider \emph{collective
  attacks}, where the adversary acts on each of the signals
independently and identically. Using the above formalism, we say that
$\cE$ is \emph{$\bar\eps$-secure against collective attacks} if $ \|
\bigl((\cE - \cF) \otimes \id\bigr)(\sigma_{\cH\cK}^{\otimes n}) \|_1
\leq \bar{\eps} $, for any (pure) $\sigma_{\cH\cK}$ on $\cH \otimes
\cK$, where $\cK \cong \cH$. This immediately implies that the same
bound holds for the extension $\tau_{\cH^n\cK^n} = \int
\sigma_{\cH\cK}^{\otimes n}d(\sigma_{\cH \cK})$ of~$\tau_{\cH^n}$,
\begin{multline} \label{eq:mixedcrit}
\|(\cE - \cF) \otimes \id_{\cK^n}  (\tau_{\cH^n\cK^n}) \|_{1} \\
 \leq \max_{\sigma_{\cH\cK}} \| (\cE - \cF) \otimes \id_{\cK^n}  (\sigma_{\cH\cK}^{\otimes n}) \|_{1} \leq \bar\varepsilon \ . 
\end{multline}

To obtain criterion~\eqref{eq:redcrit}, we need to show that a similar
bound still holds if we consider a purification $\tau_{\cH^n \cK^n
  \cN}$ of $\tau_{\cH^n\cK^n}$. For this, we think of $\cN$ as an
additional system that is available to an adversary.  Because $\cN$
can be chosen isomorphic to $\Sym^n(\cH \otimes \cK)$, its dimension
is bounded by $(n+1)^{d^2-1}$. The idea is then to compensate the
extra information available to the adversary by slightly reducing the
size of the final key. More precisely, according to the privacy
amplification theorem (Theorem~5.5.1 of~\cite{Renner05}), the protocol
$\cEp$ obtained from $\cE$ by shortening the output of the hashing by
$2 \log_2 \dim \cN \leq 2(d^2-1) \log_2 (n+1)$ bits satisfies
 \begin{align*}
   \| (\cEp - \cFp) \otimes \id_{\cK^n\cN} (\kappa) \|_{1}
 \leq
   \| (\cE - \cF) \otimes \id_{\cK^n} \otimes \tr_{\cN} (\kappa) \|_{1} \ .
 \end{align*}
 Setting $\kappa$ equal to $\tau_{\cH^n\cK^n\cN}$ and
 using~\eqref{eq:mixedcrit}, we conclude that $\| (\cE' - \cF')
 \otimes \id_{\cK^n \cN} (\tau_{\cH^n\cK^n\cN}) \|_{1} \leq \bar\eps$,
 which corresponds to~\eqref{eq:redcrit}. We have thus shown that
 $\bar{\eps}$-security of $\cE$ against collective attacks implies
 $\eps$-security of $\cE'$ against general attacks.

In the security analysis against collective attacks, the security
parameter $\bar{\eps}$ can be chosen exponentially small, i.e.,
$\bar\eps\leq 2^{-c \delta^2 n}$ (for some $c>0$), at the only cost of
reducing the key size by an (arbitrarily small) fraction $\delta$
compared to the asymptotically optimal rate. The crucial observation
made in this paper is that the security parameter $\eps$ for general
attacks and $\bar\eps$ are polynomially related
(see~\eqref{eq:redcrit}). We thus find $\eps\leq 2^{-c \delta^2
  n+(d^2-1) \log_2 (n+1)}$, which shows that security under the
assumption of collective attacks implies full security essentially
without changing the security parameter. This security estimate
improves on previous estimates based on the exponential de Finetti
theorem and has a direct impact on the security analysis of current
experimental implementations.

\emph{Concluding remark.} The technical results in this paper deal
with quantum states and channels that commute with the action of the
symmetric group on $\cH^{\otimes n}$, but can be easily generalized to
the action of an arbitrary finite or locally compact group $G$ on a
space $\cV$~\footnote{The role of $\cH^{\otimes n}$ is taken by the
  space $\cV$ of a finite-dimensional unitary representation $V$ of
  $G$.  We denote by $g\ket{v}\equiv V(g)\ket{v}$ the action of $g \in
  G$ on $\ket{v} \in \cV$ and by $g(\rho)=V(g) \rho V(g^{-1})$ the
  action on $\End(\cV)$. The space $\cK^{\otimes n}$ is replaced by a
  space $\cW\cong \cV$ on which $G$ acts with the dual representation,
  $g\ket{w}=V(g^{-1})^T \ket{w}$ for $\ket{w} \in \cW$. The role of
  the symmetric subspace $\Sym^n(\cH\otimes \cK)$ is then taken by
  $(\cV\otimes \cW)^G=\{\ket{x} \in \cV\otimes \cW: g\times g
  \ket{x}=\ket{x} \forall g \in G\}$, the \emph{invariant space} of
  $\cV\otimes \cW$. The constant $g_{n, d}$ becomes $\dim (\cV\otimes
  \cW)^G$ and the state $\tau_{\cV\cW}$ is the state proportional to
  the identity on $(\cV\otimes \cW)^G \subset \cV \otimes
  \cW$. }. Seen in the light of more general symmetry groups $G$, we
thus hope that our results will find fundamental applications in
quantum physics beyond their presented use in quantum information
theory.

\emph{Acknowledgments.} RK acknowledges support by
the NSA under ARO contract no.~W911NF-05-1-0294 and by the NSF under
contract no.~PHY-0456720. RR received support from the EU project
SECOQC.

\end{document}